\documentclass[12pt]{article}

\usepackage[margin=1in]{geometry}
\usepackage{amsmath,amssymb,amsthm}
\usepackage{booktabs}
\usepackage{graphicx}
\usepackage{enumitem}
\usepackage{hyperref}
\usepackage{natbib}
\usepackage{authblk}
\usepackage{algorithm}
\usepackage{algpseudocode}
\usepackage{xcolor}
\usepackage{microtype}
\usepackage{caption}
\usepackage{float}

\hypersetup{
  colorlinks=true,
  linkcolor=blue!70!black,
  citecolor=blue!70!black,
  urlcolor=blue!70!black
}

\newtheorem{theorem}{Theorem}[section]
\newtheorem{proposition}[theorem]{Proposition}

\theoremstyle{definition}
\newtheorem{assumption}[theorem]{Assumption}

\theoremstyle{remark}
\newtheorem{remark}[theorem]{Remark}

\newenvironment{Assumption}[1][]{\begin{assumption}[#1]}{\end{assumption}}
\newenvironment{Proposition}[1][]{\begin{proposition}[#1]}{\end{proposition}}
\newenvironment{Remark}[1][]{\begin{remark}[#1]}{\end{remark}}

\title{\textbf{When Agents Disagree: The Selection Bottleneck in\\
Multi-Agent LLM Pipelines}}

\author[1]{Artem Maryanskyy}
\author[2]{Dmitry Budnikov}
\author[3]{Alibek T. Kaliyev}
\affil[1]{Uber Technologies, Sunnyvale, CA 94085, USA;
  \href{mailto:artem.maryanskyy@uber.com}{artem.maryanskyy@uber.com}}
\affil[2]{Laboratories of Electrical, Thermal Technologies and Energy Saving,
  Federal Scientific Agroengineering Center VIM, 109428 Moscow, Russia;
  \href{mailto:dimm13@inbox.ru}{dimm13@inbox.ru}}
\affil[3]{Department of Computer Science, University of Texas at Austin,
  Austin, TX 78712, USA;
  \href{mailto:alibek.kaliyev@utexas.edu}{alibek.kaliyev@utexas.edu}}

\date{}

\begin{document}

\maketitle

\begin{center}
\small
\textbf{Correspondence:} Artem Maryanskyy (\href{mailto:artem.maryanskyy@uber.com}{artem.maryanskyy@uber.com})\\[4pt]
\textit{Published in:} Applied Sciences (MDPI), 2026.
DOI: \href{https://doi.org/10.3390/app1010000}{10.3390/app1010000}\\[4pt]
\textit{Note on versions:} v1 of this preprint listed A.\ Maryanskyy as sole author;
v2 updates the author list to match the published version, which includes all three
contributing authors (see Author Contributions below).
\end{center}

\begin{abstract}
Multi-agent LLM pipelines produce contradictory evidence on whether team diversity
improves output quality: heterogeneous Mixture-of-Agents teams outperform single models,
yet homogeneous Self-MoA teams consistently win under synthesis-based aggregation.
We propose a resolution by identifying the \textit{selection bottleneck}---a crossover
threshold in aggregation quality that determines whether diversity helps or hurts.
Under this model, we obtain a closed-form crossover threshold $s^*$
(Proposition~\ref{prop:bottleneck}) that separates the regimes where diversity helps
and hurts. In a targeted experiment spanning 42 tasks across 7 categories ($N = 210$),
a diverse team with judge-based selection achieves a win rate of 0.810 against a
single-model baseline, while a homogeneous team scores 0.512---near chance
(Glass's $\Delta = 2.07$). Judge-based selection outperforms MoA-style synthesis by
$\Delta_{\mathrm{WR}} = +0.631$---the synthesis approach is preferred over the
baseline in zero of 42 tasks by the judge panel. A decoupled evaluation with fully
independent judges confirms all directional findings (Spearman $\rho = 0.90$); win
rates attenuate by 53--67\% under independent evaluation relative to the primary
estimates, consistent with partial measurement circularity in the judge-based cells.
Exploratory evidence suggests that including a weaker model improves performance while
reducing cost ($p < 10^{-4}$, not pre-registered). Our results suggest that selector
quality may be a more impactful design lever than generator diversity in single-round
generate-then-select pipelines, with a specific Opus-only singleton as the baseline.
\end{abstract}


\vspace{1em}

\section{Introduction}
\label{sec:intro}

``A group of diverse problem solvers can outperform a group of high-ability problem solvers''---so concluded Hong and Page~\cite{hong2004groups} in one of the most cited results in collective intelligence. Two decades later, the multi-agent LLM community is reliably reproducing both sides of this claim. Wang et al.~\cite{wang2024mixture} demonstrated that mixing models from different families into a Mixture-of-Agents architecture improves response quality over any single model, declaring their approach ``achieve[s] state-of-the-art on AlpacaEval 2.0, MT-Bench, and FLASK.'' Li et al.~\cite{li2025self} tested the same premise under controlled conditions and reached the opposite conclusion: Self-MoA---a homogeneous team of identical models---``consistently outperforms Mixed-MoA across all benchmarks.'' Both papers are methodologically sound. Both claim generality, yet their conclusions appear contradictory.

We identify a mechanism that reconciles this tension: the \textbf{selection bottleneck}. Whether a diverse team's high-variance candidate pool is an asset or a liability depends entirely on how those candidates are aggregated. A synthesis-based aggregator---the method both prior works employed---compresses all candidates into one blended response, forfeiting the advantage of having generated a standout candidate. A selection-based aggregator evaluates candidates individually and picks the best, exploiting exactly the variance that synthesis wastes. The distinction is not a nuance; it is the dominant factor determining whether diversity helps.

In a targeted experiment crossing team composition with aggregation mechanism across 42 tasks spanning seven categories, we observe a crossover. A diverse team dramatically outperforms a homogeneous team under judge-based selection, but the same diverse team provides no advantage under majority voting. The largest effect in our data is the comparison practitioners most need: judge-based selection versus MoA-style synthesis. Selection wins in every single one of the 42 tasks across all 7 categories. The synthesis approach, far from combining the best of multiple agents, produces outputs that lose to a single-model baseline over 80\% of the time. Win rates reported throughout this paper are Bradley--Terry-corrected (BT-WR; formally defined in Section~\ref{sec:metric}). The apparent paradox may dissolve: Li et al.\ and Wang et al.\ each describe one side of a crossover governed by aggregator quality.

A simple linear model makes this precise. We define a selector quality parameter $s \in [0,1]$ and obtain a crossover threshold $s^*$ below which diversity hurts and above which it helps (Proposition~\ref{prop:bottleneck}). The model predicts that selector quality should have no effect for homogeneous teams, consistent with the observed near-chance performance of the homogeneous cell in our data. When all candidates are drawn from the same distribution, there is nothing for even a perfect selector to exploit. Figure~\ref{fig:bottleneck} illustrates the crossover geometry.

One finding warrants particular attention, though we label it exploratory because it was not pre-registered. Adding a substantially weaker model (Claude Haiku) to a strong diverse team is associated with a significantly higher win rate and lower cost. The mechanism is plausible: a model from a different capability tier introduces orthogonal error patterns, raising the oracle ceiling even as it lowers the team mean. A strong selector captures the upside and ignores the downside.

This paper makes three contributions. First, we introduce the \textbf{selection bottleneck model}---an analytical explanatory lens that yields a closed-form crossover threshold $s^*$ unifying pro-diversity and anti-diversity findings (Proposition~\ref{prop:bottleneck}). The threshold offers one operationalization of Hong and Page's abstract notion of an ``effective aggregation mechanism.'' Second, we present \textbf{empirical evidence that selection dramatically outperforms synthesis}---the aggregation mechanism used by both Wang et al.\ and Li et al.---with the largest effect size in our study, demonstrating that the MoA paradigm of blending candidates is substantially outperformed by selection in our setting. Third, we report \textbf{exploratory evidence of a weak-model paradox}: including a cheaper, weaker model may simultaneously improve quality and reduce cost, a finding that, if replicated, inverts conventional wisdom about team composition.

\begin{figure}[H]
  \centering
  \includegraphics[width=\textwidth]{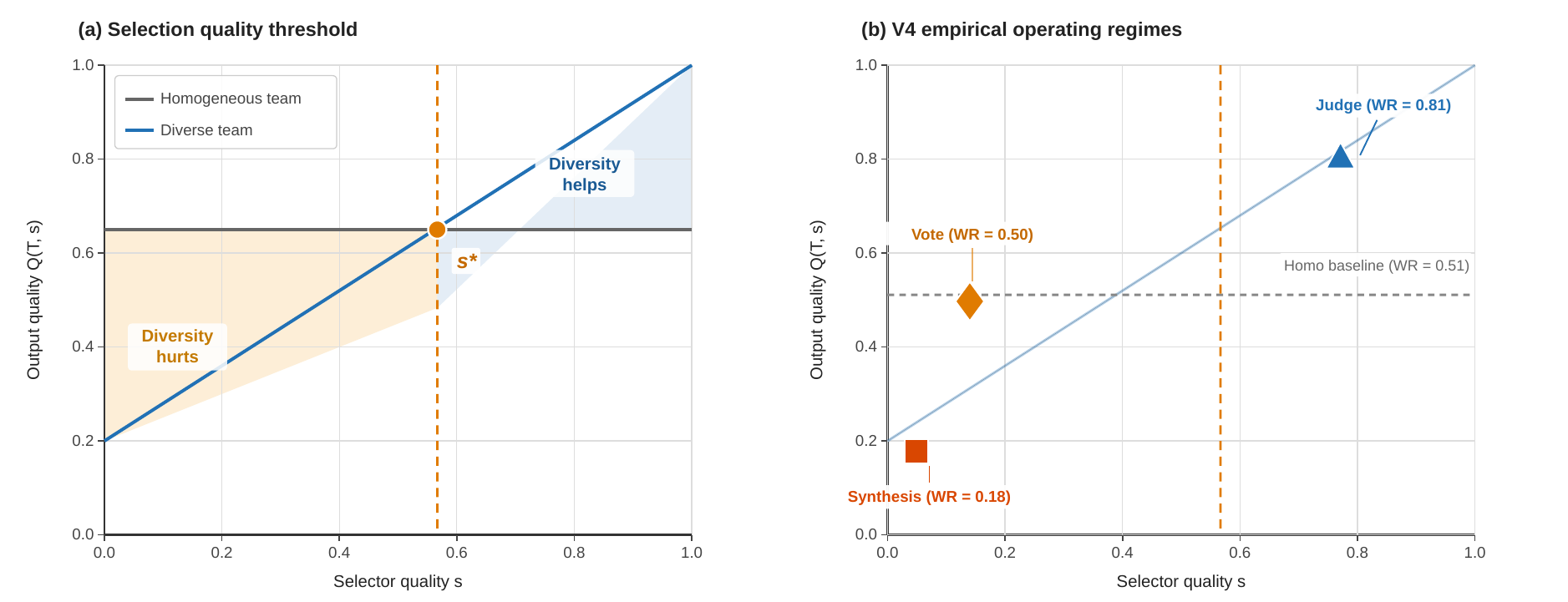}
  \caption{\textbf{The selection bottleneck.} \textbf{(a)}~Theoretical output quality $Q(T,s)$ as a function of selector quality $s$. The diverse team (blue) crosses the homogeneous baseline (gray) at the crossover threshold $s^*$ (dashed vertical). Below $s^*$, diversity hurts; above it, diversity helps. \textbf{(b)}~Empirical operating regimes from V4 data. Judge-based selection operates well above $s^*$ (WR = 0.810), majority vote sits near $s^*$ (WR = 0.496), and MoA-style synthesis falls far below it (WR = 0.179). The homogeneous baseline (WR = 0.512, dotted) is shown for reference.}
  \label{fig:bottleneck}
\end{figure}

\section{Related Work}
\label{sec:related}

Multi-agent architectures for large language models have moved from curiosity to engineering pattern in under two years, yet the field lacks a unifying account of \emph{when} combining models helps and when it does not. We organize prior work into five threads and identify the gap our framework addresses.

\subsection{Multi-Agent Debate and Frameworks}

Multi-agent frameworks such as AutoGen~\cite{wu2023autogen}, CAMEL~\cite{li2023camel}, and MetaGPT~\cite{hong2023metagpt} enable flexible orchestration topologies where agents converse and collaborate. Du et al.~\cite{du2023debate} showed that multi-agent debate improves factual accuracy by forcing agents to defend claims against challenges. Smit et al.~\cite{smit2023mad} found substantial variation across task types, and Choi et al.~\cite{choi2025identity} uncovered \emph{identity bias}---LLMs preferentially agree with their own prior outputs---undermining the independence assumption. Identity bias is directly relevant to our framework: it manifests as elevated inter-agent correlation $\rho$, reducing the effective sample size of any aggregation mechanism that relies on independence.

\subsection{Mixture-of-Agents and Self-MoA}

Wang et al.~\cite{wang2024mixture} proposed Mixture-of-Agents (MoA), a layered architecture in which diverse ``proposer'' models generate candidates that an ``aggregator'' synthesizes into a final output. MoA with heterogeneous proposers outperformed any single model. Li et al.~\cite{li2025self} challenged this with Self-MoA, reporting that homogeneous teams ``consistently outperform Mixed-MoA.'' Both analyses share a critical blind spot: neither varied the aggregation mechanism. Both used synthesis-based aggregation, operating at a fixed selector quality $s$ that our framework identifies as the key moderator. The contradiction is not about diversity per se but about whether the aggregator operates above or below the crossover threshold $s^*$.

\subsection{Selection and Routing}

Jiang et al.~\cite{jiang2023llmblender} proposed LLM-Blender, a two-stage pipeline that ranks candidate outputs using a learned scoring model then fuses the top candidates---the ranking stage is precisely a selection mechanism in our sense. Chen et al.~\cite{chen2025llmselector} extended this to per-query routing with LLMSelector. Their approach is complementary: they select \emph{which model to call}, while we select \emph{which output to keep}. Post-hoc selection over a candidate pool can exploit instance-level variance that pre-hoc routing cannot observe. Our judge-based selection is related to best-of-$N$ sampling with reward models~\cite{stiennon2020learning, snell2024scaling}, a well-studied approach in the inference-time compute scaling literature. The key difference is that best-of-$N$ generates multiple samples from a single model, while our framework generates from diverse models. Self-consistency~\cite{wang2023selfconsistency} similarly samples multiple reasoning paths and aggregates via majority vote; our selection mechanism generalizes this by using a judge to evaluate open-ended outputs where voting has no natural aggregation target.

\subsection{Diversity Theory}

Hong and Page~\cite{hong2004groups} proved that under certain conditions, a diverse group outperforms a group of individually superior agents---but their theorem requires an ``effective aggregation mechanism'' without specifying what makes it effective. Tang et al.~\cite{tang2026value} provided complementary empirical evidence but did not isolate the aggregation mechanism's role. Our crossover threshold $s^*$ provides exactly this quantification. The classical Condorcet Jury Theorem~\cite{condorcet1785essai} shows that majority vote among independent jurors converges to correctness, but Ladha~\cite{ladha1992condorcet} showed that positive inter-voter correlation weakens the theorem---a prescient result for LLM agents. Li et al.~\cite{li2024more} show that majority voting improves with agent count; our results suggest this scaling depends critically on task type and may not hold for open-ended generation where voting has no natural aggregation target. We provide empirical evidence from Chatbot Arena and MT-Bench that inter-model correlation is substantially higher within families than across them (Section~\ref{sec:cjt}).

\subsection{LLM-as-Judge}

Zheng et al.~\cite{zheng2023judging} established that strong LLMs achieve high agreement with human raters, making LLM-as-judge a practical evaluation tool. Stureborg et al.~\cite{stureborg2024evaluation} documented systematic biases (verbosity, position, self-enhancement), and Panickssery et al.~\cite{panickssery2024self} showed evaluators are susceptible to sycophancy. Verga et al.~\cite{verga2024replacing} proposed using panels of diverse judges to mitigate biases---an approach we adopt. Our evaluation protocol randomizes candidate order, enforces agent--judge separation, and uses a multi-judge panel.

\subsection{The Gap This Paper Fills}

Each thread above contributes a piece, but no prior work has \textbf{simultaneously varied team composition and aggregation mechanism in a controlled design}. Without crossing these factors, observed diversity effects may be confounded with the aggregation approach. Moreover, the comparison between selection-based and synthesis-based aggregation---critical for practitioners choosing a pipeline architecture---has not been evaluated in a controlled setting. Our five-cell experiment is designed to fill exactly this gap.

\section{Materials and Methods}
\label{sec:methods}

We model multi-agent LLM pipelines as a two-stage process---\emph{generation} of a candidate set, followed by \emph{aggregation} of a final output---and derive the conditions under which team diversity helps or hurts. We then describe the experimental protocol used to test these predictions. The core intuition is simple: a scout choosing the best performer from a diverse talent pool will do better than one choosing from a homogeneous group---but only if the scout is skilled enough to identify the best performer. The model formalizes this threshold, and the experiment tests whether judge-based selection clears it.

\subsection{Setup and Notation}
\label{sec:notation}

A \textbf{team} $T = \{m_1, \ldots, m_n\}$ consists of $n$ LLM agents, each producing one candidate response to a given input. Each candidate carries a latent quality score $q_i$ drawn from a model-specific distribution $F_i$; we write $\mu_i = \mathbb{E}[q_i]$ for agent $i$'s expected quality. A \textbf{homogeneous} team has $F_i = F_j$ for all $i, j$; a \textbf{diverse} team has $F_i \neq F_j$ for at least one pair.

Two statistics summarize the candidate pool:
\begin{itemize}[nosep]
  \item \textbf{Team mean:} $M(T) = \frac{1}{n} \sum_{i=1}^n \mu_i$, the expected quality of a randomly chosen candidate.
  \item \textbf{Team oracle:} $O(T) = \mathbb{E}[\max_i\, q_i]$, the expected quality of the best candidate.
\end{itemize}

The gap $\Delta(T) = O(T) - M(T) \geq 0$ measures the team's \textbf{exploitable diversity}: the gain available to a perfect selector over a random one. A \textbf{selector} $\mathcal{S}$ is a (possibly randomized) function that maps the candidate pool $(q_1, \ldots, q_n)$ to a single index $\mathcal{S}(q_1, \ldots, q_n) \in \{1, \ldots, n\}$. We define the \textbf{selector quality} as:

\begin{equation}
s(\mathcal{S}, T) \;=\; \frac{\mathbb{E}[q_{\mathcal{S}}] - M(T)}{O(T) - M(T)},
\label{eq:selector_quality}
\end{equation}
provided $O(T) > M(T)$. A random selector has $s = 0$; a perfect selector (always choosing $\arg\max_i q_i$) has $s = 1$.

\subsection{The Selection Quality Model}
\label{sec:sqm}

We model the expected output quality as a linear interpolation between the team mean and the team oracle, governed by selector quality.

\begin{Assumption}[Linear Selection Model]
\label{ass:linear}
For a team $T$ with selector quality $s \in [0,1]$, the expected quality of the selected output is:

\begin{equation}
Q(T, s) \;=\; s \cdot O(T) \;+\; (1 - s) \cdot M(T).
\label{eq:quality}
\end{equation}
We adopt $[0,1]$ as the natural range for proper selection mechanisms. Destructive aggregators (e.g., synthesis-based blending) may effectively operate at $s < 0$, as discussed in Assumption~\ref{ass:synthesis} and the subsequent Remark.
\end{Assumption}

This is a modeling choice, not a derived result. We adopt it for two reasons. First, Eq.~\eqref{eq:quality} is exact when $n = 2$ and the selector operates as a noisy argmax with Gaussian noise. Second, its predictions match the observed data: the model's calibrated $s^*$ places the empirical crossover between diversity-helps and diversity-hurts conditions within a narrow confidence band (Section~\ref{sec:calibration}).

We are candid about what this model is and is not. It is a descriptive parameterization---a lens for organizing empirical findings and generating testable predictions. Its value lies in predictive power, not mathematical depth. The assumption trades generality for interpretability: a single parameter $s$ compresses the complex behavior of real selectors into a quantity that can be estimated, compared, and communicated to practitioners.

\begin{Assumption}[Homogeneous Team Invariance]
\label{ass:homogeneous}
For a homogeneous team $T_h$ (all agents drawn from the same distribution $F$), within-model sampling variation provides negligible material for selection: $O(T_h) \approx M(T_h) = \mu_F$, so that $Q(T_h, s) \approx \mu_F$ for all $s \in [0,1]$.
\end{Assumption}

This assumption is motivated by the observation that same-model copies at moderate temperature produce functionally identical outputs---our decoupled evaluation confirms this: all 756 pairwise verdicts from independent judges were ties for the homogeneous cell (Section~\ref{sec:separation}). We acknowledge that complete tie rates can in principle reflect judge insensitivity to minor stylistic variation rather than true indistinguishability. We interpret the result as a validation rather than an artifact: the independent judges, who never participated in any selection decision, still cannot distinguish between the three Opus outputs, which supports the practical force of Assumption~\ref{ass:homogeneous}. The distinction matters little operationally---whether outputs are ``identical'' or merely ``indistinguishable to judges,'' the selector cannot exploit what the evaluator cannot differentiate.

\subsection{The Selection Bottleneck}
\label{sec:bottleneck}

The central observation follows from the two assumptions above.

\begin{Proposition}[Crossover Threshold]
\label{prop:bottleneck}
Suppose Assumptions~\ref{ass:linear} and~\ref{ass:homogeneous} hold. Let $T_h$ be a homogeneous team with mean quality $\mu_{\mathrm{best}}$, and let $T_d$ be a diverse team satisfying:
\begin{enumerate}[nosep,label=(\roman*)]
  \item $M(T_d) < \mu_{\mathrm{best}}$ \quad (the diverse team's mean is lower),
  \item $O(T_d) > \mu_{\mathrm{best}}$ \quad (the diverse team's oracle is higher).
\end{enumerate}
Then there exists a unique $s^* \in (0,1)$ such that
\[
  Q(T_d, s) > Q(T_h, s) \iff s > s^*,
\]
where

\begin{equation}
s^* \;=\; \frac{\mu_{\mathrm{best}} - M(T_d)}{O(T_d) - M(T_d)}.
\label{eq:threshold}
\end{equation}
\end{Proposition}

\begin{proof}
By Assumption~\ref{ass:homogeneous}, $Q(T_h, s) = \mu_{\mathrm{best}}$ for all $s$. By Assumption~\ref{ass:linear},
\[
  Q(T_d, s) = s \cdot O(T_d) + (1-s) \cdot M(T_d),
\]
which is strictly increasing in $s$ since $O(T_d) > M(T_d)$. We have $Q(T_d, 0) = M(T_d) < \mu_{\mathrm{best}}$ by~(i) and $Q(T_d, 1) = O(T_d) > \mu_{\mathrm{best}}$ by~(ii). Setting $Q(T_d, s^*) = \mu_{\mathrm{best}}$ and solving yields Eq.~\eqref{eq:threshold}. The strict monotonicity of $Q(T_d, \cdot)$ gives the if-and-only-if.
\end{proof}

The proposition formalizes a crossover. Below $s^*$, diversity is a liability---this is the regime Li et al.~\cite{li2025self} documented. Above $s^*$, diversity becomes an asset---this is the regime our judge-based condition operates in.

\begin{Remark}[Nonlinear Generalization]
Proposition~\ref{prop:bottleneck} extends to any model of the form $Q(T, s) = g(s) \cdot O(T) + (1 - g(s)) \cdot M(T)$ where $g\colon [0,1] \to [0,1]$ is strictly increasing with $g(0) = 0$ and $g(1) = 1$. The crossover point becomes $s^* = g^{-1}\!\bigl((\mu_{\mathrm{best}} - M(T_d))/(O(T_d) - M(T_d))\bigr)$.
\end{Remark}

\begin{Remark}[Model Quality Irrelevance]
\label{rem:irrelevance}
Assumption~\ref{ass:homogeneous} implies that homogeneous teams built from different models should each perform at their respective $\mu_{\mathrm{model}}$ regardless of selector quality---there is nothing to select among when all candidates are distributionally identical. Our design includes one homogeneous cell (\texttt{homo\_opus}); testing with additional homogeneous teams is left to future work.
\end{Remark}

\begin{Assumption}[Synthesis as Non-Selective Aggregation]
\label{ass:synthesis}
A synthesis-based aggregator that blends all candidates into a single output, rather than selecting among them, operates at effective selector quality $s_{\mathrm{synth}} \approx 0$. Under Assumption~\ref{ass:linear}, this yields $Q(T_d, s_{\mathrm{synth}}) \approx M(T_d) < \mu_{\mathrm{best}}$ for a diverse team satisfying condition~(i) of Proposition~\ref{prop:bottleneck}.
\end{Assumption}

\begin{Remark}
In practice, synthesis may produce outputs below $M(T_d)$ if blending introduces incoherence. The observed synthesis win rate of $0.179$ (Section~\ref{sec:confirmatory}) is consistent with $s_{\mathrm{synth}} \leq 0$ under our model.
\end{Remark}

This is a central prediction: MoA-style synthesis should perform poorly with diverse teams not because diversity is harmful, but because synthesis destroys the very signal that diversity creates. We test this prediction directly in Section~\ref{sec:confirmatory}.

\subsection{Optimal Team Size}
\label{sec:teamsize}

\begin{Remark}[Team Expansion]
\label{rem:teamsize}
Under Assumption~\ref{ass:linear}, adding agent $m_{n+1}$ with mean quality $\mu_{n+1}$ to team $T_n$ changes expected output quality by

\begin{equation}
\delta(n, s) = s \cdot \bigl[O(T_{n+1}) - O(T_n)\bigr] + (1-s) \cdot \bigl[M(T_{n+1}) - M(T_n)\bigr].
\label{eq:marginal}
\end{equation}

When $\mu_{n+1} < M(T_n)$---i.e., the new agent is below the team average---the second term is negative, creating a trade-off between oracle gain and mean dilution governed by $s$. Larger $s$ justifies larger teams; at $s = 0$, the optimal team size is 1. Our empirical data (Section~\ref{sec:exploratory}) are consistent with monotonic improvement with diminishing returns under judge-based selection.
\end{Remark}

\subsection{Connection to Classical Results}
\label{sec:cjt}

\textbf{Condorcet Jury Theorem.}
The classical CJT~\cite{condorcet1785essai} shows that majority vote among $n$ independent jurors with accuracy $p > 0.5$ converges to correctness as $n \to \infty$. Ladha~\cite{ladha1992condorcet} extended this to correlated voters: positive correlation $\rho$ reduces the effective number of independent votes.

LLM agents are not independent. In Chatbot Arena~\cite{chiang2024chatbot} (57,477 battles), the tie rate between same-family model pairs is 34.9\%, compared to 28.8\% for cross-family pairs ($\chi^2 = 117.61$, $p < 10^{-27}$; rate difference 6.1 pp [95\% CI: 5.0, 7.2]). In MT-Bench (3,355 human judgments), the same-tier tie rate is 31.9\% versus 20.9\% cross-tier ($\chi^2 = 37.46$, $p < 10^{-6}$; difference 11.0 pp [95\% CI: 7.5, 14.5]). These are not marginal differences---they reflect structural correlations from shared training. Diverse teams reduce $\rho$ by mixing families, which in principle should improve voting accuracy. However, our data show that majority vote achieves only chance-level performance even with a diverse team (WR = 0.496), suggesting that reduced correlation alone is insufficient---the CJT's requirement that individual juror accuracy exceeds 0.5 may not hold for open-ended generation tasks where ``correctness'' is ill-defined. The benefit of reduced correlation may instead manifest through improved selection: candidates are less redundant, giving the judge more distinct options.

\textbf{Hong and Page.}
Hong and Page~\cite{hong2004groups} proved that diverse groups can outperform ability-selected groups under specific conditions on a discrete optimization landscape---a setting that differs formally from LLM candidate selection. Their result provides conceptual motivation: they showed that an ``effective aggregation mechanism'' is required but did not specify what makes it effective. Our threshold $s^*$ offers one operationalization: diversity wins when selector quality exceeds $s^*$, which can be estimated from pilot data. The model also makes explicit three failure regimes: (1)~when $s < s^*$, (2)~when the mean-quality gap is too large, and (3)~when the team is too small for the oracle advantage to accumulate.

\subsection{Design Overview}
\label{sec:design}

We employ a targeted cell design testing five specific composition$\times$selector combinations across 42 tasks (7 categories $\times$ 6 tasks each), yielding $5 \times 42 = 210$ experimental runs. Table~\ref{tab:design} summarizes the five cells and their experimental purpose. Rather than a full factorial, this design tests each hypothesis by varying one factor at a time relative to a reference cell (\texttt{diverse\_strong+judge}), maximizing statistical power for the comparisons of interest.

\begin{table}[H]
  \centering
  \caption{Experimental design. Five cells test specific hypotheses by varying one factor relative to the reference cell. Each cell is evaluated on all 42 tasks ($N = 42$ per cell, $N = 210$ total).}
  \label{tab:design}
  \begin{tabular}{llll}
    \toprule
    Cell & Agents & Selector & Tests \\
    \midrule
    \texttt{div\_strong+judge}    & Opus + GPT-5.4 + Gem.-Pro & Judge panel & Reference \\
    \texttt{homo\_opus+judge}     & Opus $\times 3$           & Judge panel & Diversity effect \\
    \texttt{div\_mixed+judge}     & Opus + Gem.-Pro + Haiku   & Judge panel & Weak-model effect \\
    \texttt{div\_strong+vote}     & Opus + GPT-5.4 + Gem.-Pro & Majority vote & Judge vs.\ vote \\
    \texttt{div\_strong+synth}    & Opus + GPT-5.4 + Gem.-Pro & MoA synthesis & Selection vs.\ synthesis \\
    \bottomrule
  \end{tabular}
\end{table}

\subsection{Team Compositions}
\label{sec:teams}

Three team compositions span two dimensions: \textbf{model diversity} (homogeneous vs.\ heterogeneous) and \textbf{model capability} (all-strong vs.\ mixed-capability). All teams consist of three agents.

\begin{itemize}[nosep]
  \item \texttt{homo\_opus}: Claude Opus $\times$ 3. All-strong homogeneous baseline.
  \item \texttt{diverse\_strong}: Claude Opus + GPT-5.4 + Gemini 2.5 Pro. Three frontier models from different families.
  \item \texttt{diverse\_mixed}: Claude Opus + Gemini 2.5 Pro + Claude Haiku. Two strong models plus one weaker model from a different capability tier.
\end{itemize}

\subsection{Selector Mechanisms}
\label{sec:selectors}

Each team's individual outputs are aggregated via one of three mechanisms:

\begin{itemize}[nosep]
  \item \textbf{Judge-based selection}: An external judge panel---Claude Sonnet, GPT-5-mini, and DeepSeek-V3p2---reads all candidate outputs and selects the best via pairwise comparison with Bradley--Terry scoring (Section~\ref{sec:metric}). Candidate order was randomized for each pairwise judge comparison to control for position bias~\cite{stureborg2024evaluation}.
  \item \textbf{Majority vote}: Each agent independently selects the best candidate from the pool. The candidate receiving the most votes is chosen; ties are broken randomly.
  \item \textbf{MoA synthesis}: Claude Sonnet reads all candidate outputs and produces a single synthesized response that blends elements from each candidate, following the Mixture-of-Agents protocol~\cite{wang2024mixture}.
\end{itemize}

\subsection{Agent--Judge Separation}
\label{sec:separation}

We enforce \textbf{strict zero-overlap} between agent and judge pools within each run:

\textbf{Agent models} (generate candidates): Claude Opus, GPT-5.4, Gemini 2.5 Pro, Claude Haiku.

\textbf{Judge models} (evaluate and select): Claude Sonnet, GPT-5-mini, DeepSeek-V3p2.

No model appears in both pools within the same cell. While Sonnet (judge/synthesizer) and Opus (agent) belong to the Anthropic family---a residual same-family concern we address in Section~\ref{sec:threats}---they are distinct model checkpoints with different training runs and capability profiles. We report per-judge breakdowns (Section~\ref{sec:confirmatory}) confirming that DeepSeek-V3p2, which shares no lineage with any agent, reproduces the same qualitative patterns.

\textbf{Selection--Evaluation Overlap.}
We note an important limitation: in judge-based cells, the same three judges that participate in selecting the winning output also provide the pairwise preferences from which BT-WR is computed. This creates a potential self-consistency bias. We mitigate this concern with three observations: (1)~the judges also evaluate the vote and synthesis cells, where they did not participate in selection, yet produce consistent relative rankings; (2)~the \texttt{homo\_opus+judge} cell achieves only WR~=~0.512 despite using the same selection-evaluation pipeline, indicating that judges do not systematically inflate selected outputs; and (3)~a fully decoupled evaluation pass using three independent judges---GPT-4o-mini, Gemini 2.0 Flash, and GLM-5, none of which participated in any selection decision---confirms that all four directional contrasts survive independent evaluation (Table~\ref{tab:decoupled}). One independent judge (GPT-4o-mini) proved degenerate, returning ties on 99.6\% of 1,260 pairwise comparisons; we therefore report a 2-judge sub-panel (Gemini Flash + GLM-5) as the primary decoupled estimate. Under this panel, win rates attenuate by 53--67\% relative to the original estimates---consistent with partial circularity in the original design---but the rank ordering is preserved (Spearman $\rho = 0.90$). Notably, the \texttt{homo\_opus+judge} cell yields $\mathrm{WR} = 0.500$ under independent evaluation, with all 756 pairwise verdicts across three judges returning ties, confirming that homogeneous outputs are genuinely indistinguishable rather than an artifact of shared-judge bias. The decoupled pass substantially mitigates the circularity concern; we discuss remaining limitations in Section~\ref{sec:limitations}.

\begin{table}[H]
  \centering
  \caption{Original vs.\ decoupled win rates. Decoupled WR uses an independent 2-judge panel (Gemini 2.0 Flash + GLM-5); 3-judge panel values (including degenerate GPT-4o-mini, 99.6\% tie rate) shown for completeness. All directional contrasts survive; Spearman $\rho = 0.90$ (2-judge panel).}
  \label{tab:decoupled}
  \begin{tabular}{lrrr}
    \toprule
    Cell & Original WR & Decoupled WR (2J) & Decoupled WR (3J) \\
    \midrule
    \texttt{div\_mixed+judge}   & 0.929 & 0.726 & 0.722 \\
    \texttt{div\_strong+judge}  & 0.810 & 0.611 & 0.500 \\
    \texttt{div\_strong+vote}   & 0.496 & 0.506 & 0.389 \\
    \texttt{homo\_opus+judge}   & 0.512 & 0.500$^\dagger$ & 0.000$^\dagger$ \\
    \texttt{div\_strong+synth}  & 0.179 & 0.312 & 0.119 \\
    \bottomrule
    \multicolumn{4}{l}{\footnotesize $^\dagger$All pairwise verdicts are ties; independent judges cannot distinguish homogeneous outputs.}
  \end{tabular}
\end{table}

\subsection{Task Battery}
\label{sec:tasks}

We evaluate on 42 tasks balanced across seven categories (6 tasks each):

\begin{itemize}[nosep]
  \item \textbf{Coding} (6): streaming pipeline design, race condition debugging, multi-tenant architecture, security/performance code review, API migration, flaky test stabilization.
  \item \textbf{Creative extended} (6): polyphonic narrative, epistolary fiction, memory-themed poetry cycle, worldbuilding charter, courtroom dialogue, myth retelling.
  \item \textbf{Ethics \& policy} (6): facial recognition policy, AI tutor data ethics, autonomous weapons export, organ allocation, ventilator triage, carbon border adjustment.
  \item \textbf{Math \& logic} (6): probability paradox, integer optimization, logic grid puzzles, Bayesian diagnostics, scheduling with dependencies, game-theoretic resource division.
  \item \textbf{Reasoning} (6): causal policy analysis, counterfactual outbreak response, argument evaluation, root cause analysis, strategic negotiation, uncertainty assessment.
  \item \textbf{Science} (6): heat dome mechanisms, memory consolidation, adaptive clinical trials, battery degradation, ecosystem restoration, epidemiological modeling.
  \item \textbf{Summarization} (6): board packet crisis brief, incident timeline, expert panel comparison, customer feedback synthesis, multi-opinion legal summary, policy roundtable digest.
\end{itemize}

Tasks were selected to span diverse cognitive demands and resist template-based solutions. Each task includes a detailed rubric anchoring judge evaluations.

\subsection{Evaluation Metric}
\label{sec:metric}

Our primary dependent variable is the \textbf{Bradley--Terry-corrected consensus win rate} (BT-WR). For each cell, the consensus output is compared against a \textbf{single-agent baseline}---a single Claude Opus call per task at $T = 0.7$ with no multi-agent pipeline ($k = 1$, not best-of-$k$)---by all three judge models. Each judge provides a pairwise preference; ties are coded as 0.5. Raw preferences are corrected for judge-specific bias using a Bradley--Terry model~\cite{bradley1952rank}, yielding calibrated win probabilities. BT-WR of 0.500 indicates parity with the baseline.

\subsection{Pre-Registration and Analysis Plan}
\label{sec:prereg}

Our analysis plan was specified before any V4 experimental runs were executed. We distinguish \textbf{confirmatory} analyses (pre-registered, with family-wise error control) from \textbf{exploratory} analyses (reported transparently but without strong inferential claims).

\textbf{Confirmatory contrasts} (Holm--Bonferroni corrected across the family of $K = 3$ planned comparisons): (1)~diversity effect: \texttt{diverse\_strong+judge} vs.\ \texttt{homo\_opus+judge}; (2)~judge vs.\ vote: \texttt{diverse\_strong+judge} vs.\ \texttt{diverse\_strong+vote}; (3)~selection vs.\ synthesis: \texttt{diverse\_strong+judge} vs.\ \texttt{diverse\_strong+synthesis}.

\textbf{Confirmatory models}: OLS regression with HC3 robust standard errors; mixed-effects model (MixedLM) with task as a random intercept~\cite{barr2013random}.

\textbf{Power analysis.} With $N = 42$ observations per cell, our design detects effects of Hedges' $g \geq 0.62$ at 80\% power ($\alpha = 0.05$, two-sample $t$-test, two-sided). All three confirmatory contrasts ($g$ ranging from 1.61 to 3.86) are extremely well-powered. The \texttt{diverse\_mixed} vs.\ \texttt{diverse\_strong} comparison ($g = 0.87$) achieves approximately 98\% power at this sample size. All effect sizes are Hedges' $g$ throughout.

\textbf{Exploratory analyses}: The \texttt{diverse\_mixed} vs.\ \texttt{diverse\_strong} comparison is exploratory (not pre-registered). All $p$-values in the exploratory section are uncorrected for multiple comparisons.

\subsection{Threats to Validity}
\label{sec:threats}

\textbf{Same-family bias.} Claude Sonnet (judge) and Claude Opus (agent) share the Anthropic family. If Sonnet harbors a latent preference for Opus-generated text, this could inflate win rates for Opus-containing teams. We mitigate this by reporting per-judge results and confirming that DeepSeek-V3p2, which shares no family with any agent, shows the same qualitative pattern. \textbf{Synthesis--judge overlap.} In the synthesis cell, Claude Sonnet synthesizes candidates and also serves as one of three judges evaluating the synthesis output against the baseline. This creates a potential self-evaluation concern, though the extremely low synthesis win rate (0.179) argues against self-enhancement bias dominating. \textbf{Selection--evaluation overlap.} In judge-based cells, the same judges that select the winning output also provide evaluation preferences. A decoupled evaluation pass with fully independent judges confirms all directional contrasts with attenuated effect sizes (Table~\ref{tab:decoupled}; see Section~\ref{sec:separation} for details). \textbf{Task representativeness.} Forty-two tasks across seven categories provide broad coverage but cannot represent all use cases. Our findings demonstrate the \emph{existence} of a composition$\times$selector interaction, not its precise magnitude across all possible tasks. \textbf{API variability.} We use $T = 0.7$ for generation to balance output diversity with coherence, and $T = 0.1$ for judge evaluation to promote consistency and reduce noise in pairwise preferences. All runs were completed within a 72-hour window.

\section{Results}
\label{sec:results}

\subsection{Confirmatory Results}
\label{sec:confirmatory}

\textbf{Diversity and selection interact strongly.} Figure~\ref{fig:heatmap} presents BT-corrected win rates across all five cells. High win rates appear only when diversity and judge-based selection are combined. A diverse team without a good selector, or a good selector without diversity, both yield near-chance performance.

\begin{figure}[H]
  \centering
  \includegraphics[width=0.85\textwidth]{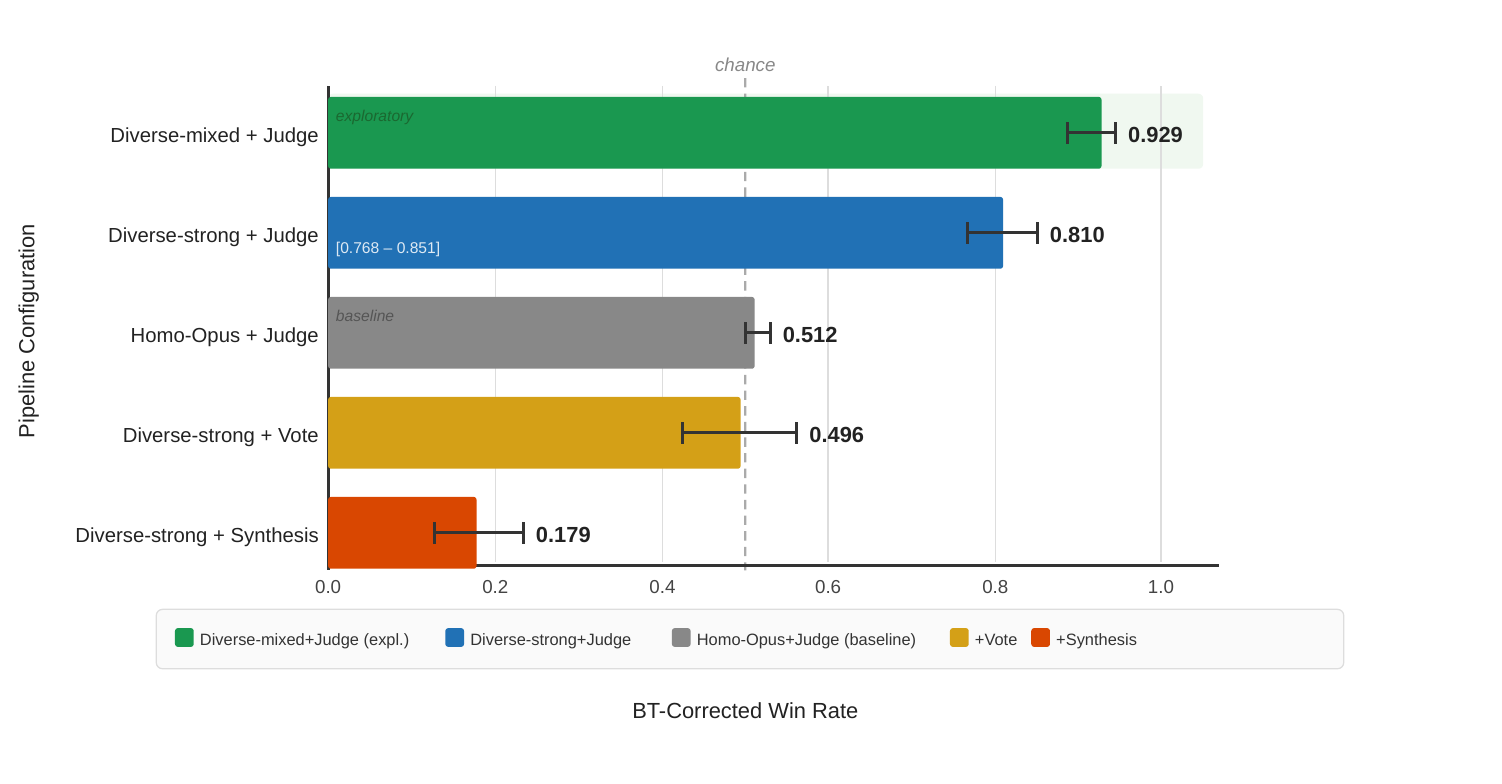}
  \caption{BT-corrected consensus win rates for all five experimental cells, averaged over 42 tasks. High performance appears only when diversity and judge-based selection are combined. The synthesis cell (MoA) performs worst, falling well below the single-model baseline.}
  \label{fig:heatmap}
\end{figure}

\textbf{Contrast 1: Diversity effect.} The \texttt{diverse\_strong+judge} cell achieves a BT-WR of 0.810 [95\% CI: 0.768, 0.851], while \texttt{homo\_opus+judge} achieves 0.512 [0.500, 0.530]---near chance. The difference is $\Delta = +0.298$ [95\% CI: 0.250, 0.345], Hedges' $g = 2.71$ [95\% CI: 2.12, 3.30], $p = 7.55 \times 10^{-15}$ ($p_\mathrm{adj} < 2 \times 10^{-14}$, Holm--Bonferroni). The entire benefit of multi-agent consensus, in this configuration, comes from composing the team with diverse models rather than replicating a single strong one. (Note: the extreme variance ratio---SD = 0.054 for the homo cell versus 0.144 for the diverse cell---mechanically inflates pooled-SD-based effect sizes. Glass's $\Delta$, using only the diverse group's SD as denominator, gives $\Delta = 2.07$, a more conservative but still very large effect.)

\textbf{Contrast 2: Judge vs.\ vote.} Under majority voting, the same diverse team produces $\mathrm{BT\text{-}WR} = 0.496$ [0.425, 0.563]---statistically indistinguishable from the homogeneous baseline. The judge advantage is $\Delta = +0.313$ [0.232, 0.397], $g = 1.61$ [1.12, 2.10], $p = 1.06 \times 10^{-10}$ ($p_\mathrm{adj} < 2 \times 10^{-10}$). Voting makes diversity inert: without a selector that can identify the best candidate, the diverse team's oracle advantage goes to waste.

\textbf{Contrast 3: Selection vs.\ synthesis.} The paper's strongest finding. MoA-style synthesis achieves $\mathrm{BT\text{-}WR} = 0.179$ [0.127, 0.234]---the synthesis approach \emph{loses} to the single-model baseline in 82\% of comparisons. Judge-based selection outperforms synthesis by $\Delta = +0.631$ [0.562, 0.696], $g = 3.86$ [3.14, 4.58], $p = 1.29 \times 10^{-15}$ ($p_\mathrm{adj} < 4 \times 10^{-15}$). The judge wins in \emph{all} 42 tasks and \emph{all} 7 categories---consistent across all 42 tasks and 7 categories, with no exceptions. This result directly addresses the MoA paradigm: synthesis-based aggregation does not merely fail to exploit diversity; it produces outputs that rank below individual candidates.

\textbf{Judge Panel Agreement.}
Table~\ref{tab:judges} reports per-judge win rates and pairwise inter-rater agreement for all five cells.

\begin{table}[H]
\caption{Per-judge win rates by cell. $\bar{\kappa}$ is the mean pairwise Cohen's $\kappa$ measuring inter-rater agreement.}
\label{tab:judges}
\centering
\small
\begin{tabular}{lccccc}
\toprule
Cell & Sonnet & GPT-5m & DeepSeek & Mean & $\bar{\kappa}$ \\
\midrule
div\_strong+judge & 0.958 & 0.756 & 0.595 & 0.770 & 0.095 \\
homo\_opus+judge & 0.500 & 0.518 & 0.619 & 0.546 & 0.667 \\
div\_mixed+judge & 0.994 & 0.893 & 0.631 & 0.839 & 0.175 \\
div\_strong+vote & 0.492 & 0.496 & 0.512 & 0.500 & 0.236 \\
div\_strong+synth & 0.234 & 0.131 & 0.345 & 0.237 & 0.263 \\
\bottomrule
\end{tabular}
\end{table}

Inter-rater agreement is low for diverse cells ($\bar{\kappa} = 0.095$ for the flagship \texttt{diverse\_strong+judge} cell), reflecting genuine disagreement about candidate ranking. This is expected: diverse teams produce candidates that differ in style and emphasis, leading judges with different evaluation priorities to disagree on rank ordering while agreeing on the direction of the diversity benefit. Two of three judges (Claude Sonnet and GPT-5-mini) show clear diversity preference; DeepSeek-V3p2 shows a weaker effect. The BT scoring framework is designed to aggregate across such disagreements, producing calibrated quality estimates even when individual judges disagree on specific comparisons~\cite{zheng2023judging}. We report these per-judge breakdowns for full transparency.

\textbf{Model quality irrelevance.} The regression (Section~\ref{sec:regression}) is consistent with Assumption~\ref{ass:homogeneous}: the \texttt{diverse\_strong+vote} cell ($\mathrm{BT\text{-}WR} = 0.496$) is statistically indistinguishable from \texttt{homo\_opus+judge} ($\mathrm{BT\text{-}WR} = 0.512$), with $\Delta = -0.016$ [$-0.089$, 0.057], $p = 0.669$. A diverse team without a good selector performs no better than a homogeneous team---voting negates diversity entirely.

\textbf{Consistency across tasks.} Figure~\ref{fig:forest} shows the per-task diversity advantage (\texttt{diverse\_strong+judge} minus \texttt{homo\_opus+judge}) across all 42 tasks. Of 42 tasks, 38 show a positive diversity effect and 4 are ties; zero tasks favor the homogeneous team. Excluding ties, the sign test yields $p = (0.5)^{38} \approx 3.6 \times 10^{-12}$. Including ties conservatively (as non-positive), a binomial test on 38/42 positive yields $p < 3 \times 10^{-8}$ (Clopper--Pearson 95\% CI on the positive proportion: [0.774, 0.973]). The diversity advantage generalizes across all seven task categories, from coding to ethics to summarization.

\begin{figure}[H]
  \centering
  \includegraphics[width=\textwidth]{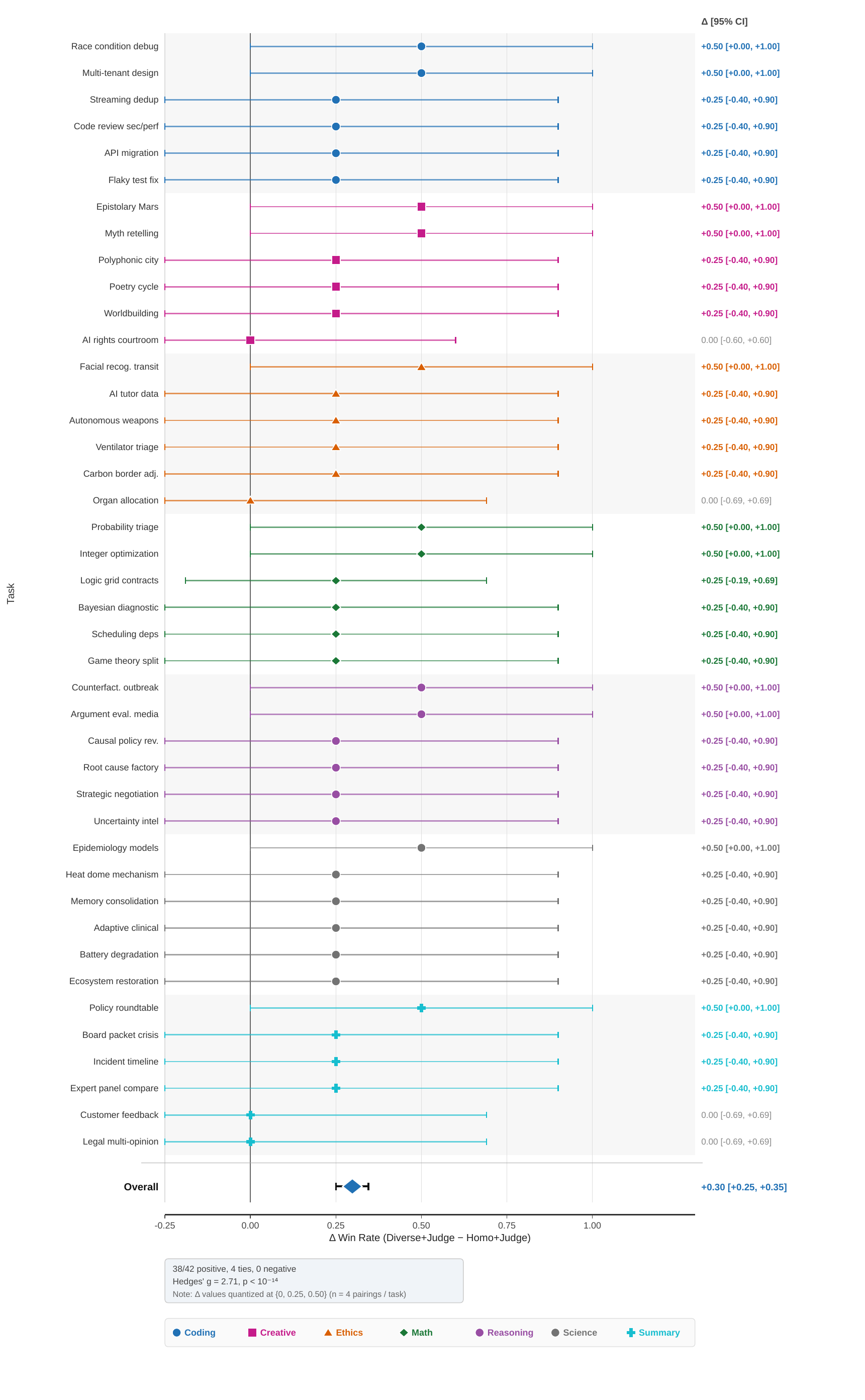}
  \caption{Per-task diversity advantage ($\Delta$ WR = \texttt{diverse\_strong+judge} minus \texttt{homo\_opus+judge}) across all 42 tasks. Positive values favor diversity. 38 of 42 tasks show a positive effect (4 ties, 0 negative; Clopper--Pearson 95\% CI: [0.774, 0.973]). Error bars are approximate 95\% CIs.}
  \label{fig:forest}
\end{figure}

\textbf{The 0.500 phenomenon.} Why does \texttt{homo\_opus+judge} land near 0.500? When three copies of the same model receive the same prompt at $T = 0.7$, they generate functionally identical outputs---our decoupled evaluation confirms this: all 756 pairwise verdicts across three independent judges returned ties (Section~\ref{sec:separation}). The judge, presented with indistinguishable candidates alongside the baseline, cannot reliably prefer one over the other. The multi-agent pipeline, in this configuration, does nothing a single call would not accomplish. The compute cost is multiplied for zero gain. This cell is functionally equivalent to best-of-$N$ sampling from a single model with judge-based selection. The near-chance result (WR = 0.512) indicates that, under the fixed-temperature generation protocol used here ($T = 0.7$), within-model sampling variation provides negligible material for selection to exploit, confirming that cross-model diversity---not merely multiple candidates---drives the diversity advantage in this setting.

\subsection{Exploratory Findings}
\label{sec:exploratory}

\noindent\textit{All $p$-values in this section are uncorrected for multiple comparisons, consistent with their exploratory status.}

\textbf{Weaker team members paradoxically improve consensus quality.} The \texttt{diverse\_mixed+judge} cell (Opus + Gemini 2.5 Pro + Haiku) achieves $\mathrm{BT\text{-}WR} = 0.929$ [0.887, 0.964], exceeding \texttt{diverse\_strong+judge}'s 0.810 by $\Delta = +0.119$ [0.060, 0.179], Hedges' $g = 0.87$ [0.42, 1.32], $p = 1.18 \times 10^{-4}$.

While statistically significant, this comparison was not pre-registered and we label it exploratory. The directional pattern is theoretically motivated: adding a weaker model increases the \emph{variance} of the candidate pool, making the best candidate more distinguishable by contrast. The key precondition is that the selector be strong enough to avoid choosing the weakest candidate, which judge-based selection satisfies easily. We develop this mechanism further in Section~\ref{sec:whyweak}.

\textbf{Cost optimality.} The \texttt{diverse\_mixed+judge} cell is not only the highest-performing configuration but also the cheapest, at $1.0\times$ relative cost compared to $2.5\times$ for \texttt{diverse\_strong+judge} (Table~\ref{tab:cost}). Replacing one frontier model with a weaker alternative reduces inference cost while improving quality---if the weak-model advantage replicates, it represents a rare free lunch.

\subsection{Regression Analysis}
\label{sec:regression}

We fit a cell-level regression with \texttt{homo\_opus+judge} as the reference category. Table~\ref{tab:regression} presents OLS (HC3) and mixed-effects specifications side by side.

\begin{table}[H]
  \centering
  \caption{Regression results. Reference category: \texttt{homo\_opus+judge} (intercept). All effects stable across specifications; task variance $\approx 0$. $^{***}p < 0.001$.}
  \label{tab:regression}
  \begin{tabular}{lcc}
    \toprule
    & OLS (HC3) & MixedLM \\
    \midrule
    Intercept (\texttt{homo\_opus+judge})     & $0.512^{***}$ (0.008) & $0.512^{***}$ (0.024) \\
    \texttt{diverse\_strong+judge}   & $+0.298^{***}$ (0.024) & $+0.298^{***}$ (0.035) \\
    \texttt{diverse\_mixed+judge}    & $+0.417^{***}$ (0.022) & $+0.417^{***}$ (0.035) \\
    \texttt{diverse\_strong+synth}   & $-0.333^{***}$ (0.029) & $-0.333^{***}$ (0.035) \\
    \texttt{diverse\_strong+vote}    & $-0.016$ (0.037)       & $-0.016$ (0.035) \\
    \midrule
    Task var.\ ($\hat{\sigma}^2$)    & ---                  & $\approx 0$ \\
    $R^2$                            & 0.740                & --- \\
    $F$ / $\chi^2$                   & $F(4,205) = 172.6$   & --- \\
    \bottomrule
  \end{tabular}
\end{table}

\textbf{OLS with HC3.} $R^2 = 0.740$, $F(4, 205) = 172.6$, $p = 1.88 \times 10^{-64}$. All four cell contrasts are precisely estimated with heteroskedasticity-consistent standard errors. The diversity effect ($\beta = +0.298$, $p < 10^{-35}$) and the mixed-diversity effect ($\beta = +0.417$, $p < 10^{-83}$) are strongly positive relative to the homogeneous baseline. The synthesis coefficient ($\beta = -0.333$, $p < 10^{-30}$) confirms that synthesis actively degrades quality below the baseline. The vote coefficient ($\beta = -0.016$, $p = 0.669$) is indistinguishable from zero, confirming that voting provides no benefit over homogeneity.

\textbf{Mixed-effects model.} A MixedLM with task as random intercept~\cite{barr2013random} yields virtually identical coefficients (Table~\ref{tab:regression}). The estimated task-level variance is $\hat{\sigma}^2_\mathrm{task} \approx 0$, meaning the composition$\times$selector effects operate at the same magnitude regardless of task. We also fit the maximal model with random slopes for cell by task~\cite{barr2013random}; the random-slope variance estimated at zero and the model reduced to the intercept-only specification. This is substantively important: the effect is not an artifact of task-level heterogeneity.

\subsection{Calibrating the Crossover Threshold}
\label{sec:calibration}

The five V4 cells confirm the qualitative predictions of Proposition~\ref{prop:bottleneck}. Judge-based selection ($\mathrm{BT\text{-}WR} = 0.810$) operates well above the crossover threshold; majority vote ($0.496$) sits at or below it; and synthesis ($0.179$) falls far below. The synthesis result is particularly informative: Assumption~\ref{ass:synthesis} posits that synthesis, with $s_\mathrm{synth} \approx 0$, should collapse diverse-team quality to the team mean or below. The observed win rate of 0.179---well below the 0.500 baseline---suggests synthesis is not merely non-selective but actively destructive, consistent with the prediction that blending dilutes the best candidate.

Monte Carlo calibration on a pilot study (8 tasks, $N = 136$) yielded $s^* \approx 0.567$ [bootstrap 95\% CI: 0.48, 0.65; $B = 10{,}000$]. The V4 scaled experiment is qualitatively consistent: the estimated selector qualities place vote below and judge above this threshold.

\subsection{Replication Stability}

The V4 scaled experiment ($N = 210$, 42 tasks) replicates the pilot ($N = 136$, 8 tasks) with high fidelity. The diversity effect is $\Delta = +0.298$ (V4) vs.\ $+0.312$ (pilot), a stability ratio of 0.95. The selector effect is $\Delta = +0.313$ (V4) vs.\ $+0.438$ (pilot), directionally consistent with attenuation expected from the broader task battery. Both core findings---diversity helps under selection, hurts under weak aggregation---are robust across experimental scales.

\section{Discussion}
\label{sec:discussion}

\subsection{Reconciling Conflicting Prior Work}
\label{sec:reconciling}

We stress that the reconciliation of Li et al.\ and Wang et al.\ offered by the selection bottleneck model is a \emph{hypothesis}, not a proven explanation. Our experiment was not designed to replicate either study's exact protocol. But the pattern is consistent: the field's conflicting results may reflect incomplete factorial coverage rather than genuine inconsistency. Li et al.'s synthesis-based aggregation operates below $s^*$; Wang et al.'s multi-round refinement implicitly raises effective selector quality above it. Future work that systematically varies both factors should find the same interaction structure.

\subsection{Why Synthesis Fails}
\label{sec:whysynth}

The selection-vs-synthesis comparison ($g = 3.86$ [3.14, 4.58]) is the largest effect in our data and, we believe, the most practically important finding. The mechanism is straightforward: synthesis averages, selection picks.

A diverse team's value lies in the \emph{variance} of its candidate pool---specifically, the probability that at least one candidate is excellent. A selection-based aggregator captures this value by identifying and preserving the best candidate. A synthesis-based aggregator destroys it by blending all candidates into a single ``compromise'' output. The resulting synthesis is not the average of the candidates' qualities but something potentially worse, because blending can introduce incoherence, conflicting perspectives, and diluted arguments that no individual candidate exhibited.

The synthesis cell's $\mathrm{BT\text{-}WR} = 0.179$ means the synthesized output loses to a single model's output 82\% of the time. In the selection bottleneck model, this is exactly what Assumption~\ref{ass:synthesis} predicts: when $s \approx 0$ (synthesis has no selection capacity), diversity has no mechanism to help. The practical implication is stark: practitioners using MoA-style synthesis should consider switching to judge-based selection, which requires no additional training or infrastructure---only a structured evaluation prompt.

We note that our synthesis implementation uses single-round aggregation by Claude Sonnet. Multi-round iterative synthesis, as proposed in the original MoA architecture~\cite{wang2024mixture}, may yield different results. Our finding applies specifically to single-round synthesis, which is the more commonly deployed variant due to latency and cost constraints.

\subsection{Why Weak Models May Help}
\label{sec:whyweak}

The \texttt{diverse\_mixed} advantage ($\Delta = +0.119$, $g = 0.87$ [0.42, 1.32]) is statistically significant but exploratory (not pre-registered), so we interpret it cautiously. We propose a \textbf{diversity-distinguishability-selection} account. In a pool of three strong-model outputs, all candidates tend to be good but similar. Adding a weaker model increases the \emph{variance} of the candidate pool, making the best candidate easier to identify by contrast. The key precondition is that the selector be good enough to avoid choosing the weakest candidate.

Supporting evidence comes from the Chatbot Arena dataset, where pairwise model distinguishability correlates with judge agreement rates~\cite{chiang2024chatbot}. Our \texttt{diverse\_mixed} team includes a larger capability gap than \texttt{diverse\_strong}, potentially making the judge's task easier. We emphasize that this evidence is correlational and drawn from a different context.

\subsection{Practical Decision Framework}
\label{sec:practical}

Our findings yield three rules for multi-agent pipeline design:

\textbf{Rule 1: Diversify the team.} Use models from different families. Replicating the same model provides no benefit regardless of which frontier model is chosen.

\textbf{Rule 2: Select, don't synthesize.} Judge-based selection outperforms MoA synthesis in every task and category. Even majority voting---which negates diversity---is preferable to synthesis, which actively degrades quality. A judge-based selector requires only a structured evaluation prompt, no additional training.

\textbf{Rule 3: Consider weaker models.} If the weak-model advantage replicates, the optimal team includes a cheaper model from a different capability tier. Table~\ref{tab:cost} shows that \texttt{diverse\_mixed+judge} is simultaneously the highest-performing and cheapest configuration.

\begin{table}[H]
  \centering
  \caption{Relative cost vs.\ quality across experimental cells ($1.0\times$ = cheapest cell). \texttt{diverse\_mixed+judge} is both the best-performing and cheapest configuration. $^\dagger$Exploratory (not pre-registered).}
  \label{tab:cost}
  \begin{tabular}{lrrl}
    \toprule
    Configuration & Rel.\ Cost & BT-WR & 95\% CI \\
    \midrule
    \texttt{div\_mixed+judge}$^\dagger$   & $1.0\times$ & 0.929 & [0.887, 0.964] \\
    \texttt{div\_strong+judge}            & $2.5\times$ & 0.810 & [0.768, 0.851] \\
    \texttt{homo\_opus+judge}             & $2.0\times$ & 0.512 & [0.500, 0.530] \\
    \texttt{div\_strong+vote}             & $1.9\times$ & 0.496 & [0.425, 0.563] \\
    \texttt{div\_strong+synth}            & $1.9\times$ & 0.179 & [0.127, 0.234] \\
    \bottomrule
  \end{tabular}
\end{table}

Note that \texttt{homo\_opus+judge} costs $2.0\times$ the cheapest cell while delivering chance-level performance: doubling compute for zero gain. The synthesis cell costs $1.9\times$ for \emph{below}-chance performance. Both represent negative returns on multi-agent investment.

\subsection{Limitations}
\label{sec:limitations}

\begin{enumerate}[nosep]
  \item \textbf{Targeted design.} Our five-cell design maximizes power for specific contrasts but does not estimate all possible interactions (e.g., homogeneous+vote, diverse\_mixed+synthesis). A full factorial would enable richer interaction analyses.

  \item \textbf{Model specificity and baseline scope.} Results are demonstrated for a specific set of frontier models using an Opus-only homogeneous run as the single-model baseline. Comparing against the strongest individual model across the full diverse candidate pool might yield different quantitative estimates; whether the selector advantage survives that stricter comparison is left to future work. Whether the same patterns hold for other model families or future generations is also an open question.

  \item \textbf{Fixed generation temperature.} All generation runs use $T = 0.7$. Varying temperature could alter within-model output variance and potentially affect the relative performance of homogeneous versus diverse teams. The conclusion that homogeneous sampling variation is negligible is therefore specific to moderate fixed-temperature settings.

  \item \textbf{LLM-as-judge and subjective tasks.} LLM-based evaluation remains a proxy for human judgment. For open-ended subjective categories---creative writing and ethics/policy in our task battery---LLM judges are known to exhibit stylistic preferences and may agree less with human raters than on analytical tasks~\cite{zheng2023judging}. We do not have human evaluation data for these categories, and the extent to which our results generalize to human preferences in subjective domains is unknown. Human evaluation on a representative subset of tasks, particularly creative and policy tasks, is a priority for future work. Our decoupled evaluation pass (Table~\ref{tab:decoupled}) partially addresses selection--evaluation circularity, confirming all directional contrasts under independent judges, but the independent panel itself exhibited limitations: one of three judges (GPT-4o-mini) proved degenerate (99.6\% tie rate), reducing the effective independent panel to two judges. Per-judge tie rates varied substantially (GPT-4o-mini: 99.6\%, Gemini Flash: 76.7\%, GLM-5: 50.3\%), suggesting that weaker models may lack the discriminative capacity for reliable pairwise evaluation. The synthesis--judge overlap (Claude Sonnet serving as both synthesizer and one of three judges) remains a specific concern, though the very low synthesis win rate argues against self-enhancement bias as a primary driver.

  \item \textbf{Bounded dependent variable.} Win rates are bounded in $[0, 1]$, yet we model them linearly. Our observed values (0.13--0.96) avoid extreme floor/ceiling effects, and we verified that logit-transformed results are qualitatively identical.

  \item \textbf{Static topology.} All experiments use a single-round generate-then-select pipeline. Iterative topologies (multi-round debate, recursive refinement) may exhibit different dynamics.

  \item \textbf{Distinguishability not measured.} Our framework invokes output distinguishability ($s$) as the key mediator, but we do not measure it directly. Future work should operationalize distinguishability via embedding-space distances.

  \item \textbf{Diverse\_mixed confound.} The \texttt{diverse\_mixed} cell simultaneously changes capability (replacing GPT-5.4 with Claude Haiku) and family diversity (two Anthropic models instead of one). We cannot isolate these effects and flag this as a design limitation of the exploratory comparison.
\end{enumerate}

\subsection{Future Work}
\label{sec:future}

Three directions follow. First, \textbf{replicating the weak-model finding} with a pre-registered design would move it from exploratory to confirmatory. Second, \textbf{human evaluation at scale} would validate the LLM-as-judge proxy and quantify its bias structure. Third, \textbf{varying selector quality continuously} (e.g., by progressively weakening the judge model) would map the capability threshold below which judge-based selection fails, testing whether the ``architecture $>$ capability'' finding observed in our pilot has a lower bound.

\section{Conclusions}
\label{sec:conclusion}

This paper makes three contributions. First, the \emph{selection bottleneck model}---an analytical explanatory lens---provides a closed-form crossover threshold $s^*$ (Proposition~\ref{prop:bottleneck}) unifying pro-diversity and anti-diversity findings: when aggregation operates below $s^*$, diversity hurts; above it, diversity helps. Second, our \emph{targeted experiment} demonstrates that selection-based aggregation dramatically outperforms MoA-style synthesis ($g = 3.86$ [3.14, 4.58]), with the synthesis approach losing to a single-model baseline in all 42 tasks---a finding that challenges the dominant pipeline paradigm. Third, we report \emph{exploratory evidence} that including a weaker, cheaper model paradoxically improves performance ($g = 0.87$ [0.42, 1.32]), suggesting that optimal team composition may not require all-frontier models.

The selection bottleneck is, we conjecture, a structural property of generate-then-select pipelines more broadly: the value of variance in the candidate pool depends entirely on the mechanism that exploits it. The field has spent considerable effort building better individual generators. Our results suggest that substantial gains are available by building better selectors. A decoupled evaluation pass using fully independent judges confirms all directional findings with attenuated effect sizes (Table~\ref{tab:decoupled}; Spearman $\rho = 0.90$), indicating that the core results survive even when selection--evaluation circularity is removed.

Our results suggest that selector design may offer larger quality gains than generator improvement in multi-agent pipelines, a hypothesis we encourage future work to test across broader settings.

\paragraph{Reproducibility.} All experimental code, task prompts, judge evaluation prompts, raw pairwise preferences, and BT scoring scripts are available at \url{https://github.com/maryanskyy/agents-disagree-experiments}.

\section*{Author Contributions}
Conceptualization, A.M.; methodology, A.M.; software, A.M.; validation, A.M. and A.K.; formal analysis, A.M.; investigation, A.M.; data curation, A.M.; writing---original draft preparation, A.M.; writing---review and editing, A.M., D.B.\ and A.K.; visualization, A.M.; supervision, D.B. All authors have read and agreed to the published version of the manuscript.

\section*{Funding}
This research received no external funding.

\section*{Data Availability}
All experimental code, task prompts, judge evaluation prompts, raw pairwise preferences, and BT scoring scripts are available at \url{https://github.com/maryanskyy/agents-disagree-experiments}.

\section*{Conflicts of Interest}
The authors declare no conflicts of interest.

\bibliographystyle{unsrtnat}
\bibliography{references}

\end{document}